\documentclass[11pt, oneside]{article} 
\usepackage[margin=2cm]{geometry}                		
\usepackage[parfill]{parskip}    		
\usepackage{graphicx}			
\usepackage[utf8]{inputenc}
\usepackage[english]{babel} 
\usepackage{amsmath,amssymb,amsthm,bm} 
\usepackage[margin=2cm]{geometry}
\usepackage[color=yellow]{todonotes}
\usepackage{mathrsfs}
\usepackage{url}	
\usepackage{esint}
\usepackage[colorlinks]{hyperref}
\usepackage{enumerate}
\usepackage{outlines}[enumerate]
\usepackage{framed,comment,enumerate}
\usepackage{upgreek}
\usepackage{pgfplots}
\usepackage{float}
\usepackage{tikz,tikz-cd}
\usepackage{rotating}
\usepackage{graphicx}
\usepackage{caption}
\usepackage{multicol}
\usepackage{cancel}
\usepackage{subcaption}
\usepackage[title]{appendix}
\usepackage{tikz-cd}
\usepackage{pgf, pgffor}
\extrafloats{100}
\numberwithin{equation}{section}

\newtheorem{lemma}{Lemma}[section]

\newtheorem{remark}{Remark}[section]

\theoremstyle{definition}

\newcommand{\mc}[1]{\mathcal{#1}}
\newcommand{\bs}[1]{\boldsymbol{#1}}

\newcommand{\mcal}[1]{\mc{#1}}
\newcommand{\scp}[2]{\left<#1\,,\,#2\right>}

\newcommand{\ad}{\operatorname{ad}}

\def\p{{\partial}}

\def\rmd{{\color{red}{\rm d}}}

\def\bA{{\mathbf{A}}}

\def\bk{{\mathbf{k}}}
\def\bm{{\mathbf{m}}}

\def\bu{{\mathbf{u}}}

\def\bx{{\mathbf{x}}}

\def\p{\partial}

\pgfplotsset{compat=1.16}

\def\bk{\mathbf{k}}

\def\bm{\mathbf{m}}
\def\bu{\mathbf{u}}

\def\bx{\mathbf{x}}

\newcommand\fder[2]{\frac{\delta {#1}}{\delta {#2}}}
\newcommand\lieder[2]{\mathcal{L}_{#1} #2}

\def\contract{\makebox[1.2em][c]{\mbox{\rule{.6em}
{.01truein}\rule{.01truein}{.6em}}}}

\begin{document}

\title{\textbf{Collisions of Burgers Bores with Nonlinear Waves} }
\author{A. Dombret$^1$, D.D. Holm$^2$, R. Hu$^2$, O.D. Street$^2$, H. Wang$^2$\footnote{Corresponding author, email: hanchun.wang21@imperial.ac.uk} \\ \footnotesize
(1) Physics, \'Ecole Normale Sup\'erieure,
(2) Mathematics, Imperial College London \\ \footnotesize
albert.dombret@ens.psl.eu, d.holm@ic.ac.uk, ruiao.hu15@imperial.ac.uk,\\\footnotesize
o.street18@imperial.ac.uk, hanchun.wang21@imperial.ac.uk
\\ \small
}
\date{}

\maketitle

\begin{abstract}
This paper treats nonlinear wave-current interactions in their simplest form -- as an overtaking collision. In one spatial dimension, the paper investigates the collision interaction formulated as an initial value problem of a Burgers bore overtaking solutions of two types of nonlinear wave equations -- Korteweg–de Vries (KdV) and nonlinear Schr\"odinger (NLS). The bore-wave state arising after the overtaking Burgers-KdV collision in numerical simulations is found to depend qualitatively on the balance between nonlinearity and dispersion in the KdV equation. The Burgers-KdV system is also made stochastic by following the stochastic advection by Lie transport approach (SALT). 
\end{abstract}


\tableofcontents


\section{Introduction}\label{sec: Intro}

Our topic is the nonlinear momentum exchange between surface waves and the currents which carry them. For example, the wind stress creates waves and swells on the sea surface. Those sea surface waves may then exchange momentum with the fluid flow at the surface. We model the last step via a composition-of-maps approach in which waves are taken as a degree of freedom which exchanges momentum and energy with the fluid current that carries them. 

Our composition-of-maps approach models nonlinear surface waves as propagating in the reference frame of the fluid velocity at the surface. The total momentum of the system may be written as the sum of wave momentum and fluid momentum in the fixed Eulerian frame. By Newton’s 2nd Law, the time rate of change of this total momentum equals the true force acting on the fluid flow in an inertial frame. However, the acceleration of the fluid velocity is only part of the rate of change of this total force. Thus, the acceleration of the fluid velocity appears to acquire an additional fictitious force (such as the Coriolis force, or Craik-Leibovich force) when experienced in the non-inertial reference frame of the fluid velocity. 

The non-inertial force in the fluid frame arising from the shift of the fluid momentum by the addition of the wave momentum in the Eulerian frame can be regarded as an additional  source of circulation around a Lagrangian loop carried by the fluid velocity. As with the Coriolis force or the Craik-Leibovich vortex force, the difference in fluid velocity circulation dynamics induced by measuring the fluid velocity circulation relative to the wave velocity circulation can be exhibited by calculating the Kelvin theorem for the full system and subtracting out the wave velocity contribution to the circulation integral. To treat this wave-current momentum interaction dynamics, a theory based on the composition of the fluid flow map and the wave dynamics map has been developed recently in \cite{HHS2022a,HHS2023}.

This wave-current momentum interaction dynamics can be illustrated in 1D by setting up an initial condition for a wave-current `collision’ in which, for example, a Burgers ramp/cliff solution (an advancing B-bore velocity profile of the fluid current) overtakes a set of KdV nonlinear dispersive wave packets in its path. The numerical simulations in the present work show for these initial conditions that when the B-bore overtakes the KdV wave packets, the large fluid velocity gradient at the leading edge of the B-bore can rapidly feed the amplitude of the KdV waves so much that -- in the frame of motion of the Burgers leading edge -- the KdV solution can incorporate part of the Burgers velocity and carry it forward ahead of the new Burgers leading edge as a compound wave. 


Physically, though, real bores driven by the tide and advancing up the Severn river for example tend not to show this numerically simulated formation of compound waves. Instead, they tend to show a train of small amplitude surface waves which have been swept up and embedded in the shallow ramp profile behind the advancing front of the bore \cite{CotterBokhove2009}. This qualitative difference in behaviour is found in numerical simulations in the present work to depend on the balance between nonlinearity and dispersion in the KdV equation. Namely, as the dispersion coefficient is raised at fixed nonlinearity coefficient the behaviour of the KdV waves in the B-KdV system can switch their behaviour from being passively incorporated into the Burgers velocity profile to actively incorporating part of the Burgers momentum and breaking away run ahead of the Burgers front as a compound wave. See Figure \ref{fig:gamma1_2} for a comparison of simulation results.   

Thus, from the viewpoint of the present work, this different behaviour in our simulations of the B-KdV system arises because of a bifurcation depending on the balance between the KdV nonlinearity and the KdV dispersion and perhaps also with the Burgers nonlinearity. This type of bifurcation study is in progress also for the wave-current interaction of Burgers currents and waves governed by the nonlinear Schrodinger (NLS) equation. However, a discussion of the study for B-NLS collisions will be deferred to future work. For the treatment of wave-current interaction between NLS waves carried by Euler fluid motion in two-dimensions, see \cite{HHS2022a,HHS2023}.

\section{Modelling considerations}
Previous work has shown that in models of wave mean flow interaction (WMFI), although the mean flow may not itself create waves, the interaction of the mean flow with existing waves can have strong effects both on the mean flow and on the waves, \cite{HHS2023}. In this paper, we will exhibit a geometric approach to wave-current interaction based upon composition of maps introduced in \cite{HHS2023}, which we illustrate by considering two examples of a bore -- whose dynamics is governed by the Burgers equation -- overtaking a set of water waves governed either by the Korteweg-de Vries (KdV) equation in one example, or by the nonlinear Schr\"odinger (NLS) equation in the other.  The dynamics of all three of these types of water waves are well known. However, the application of the method of composition of flow maps to the collision interaction of a Burgers bore overtaking a set of nonlinear shallow water waves seems to be new. 

The present work aims to investigate nonlinear wave-current interactions in their simplest forms, in one dimension (1D) on the real line $\mathbb{R}$. Even in 1D these interactions of different types of waves can be profound. In particular, we investigate overtaking interactions of ramps-and-cliffs shaped bore solutions of the inviscid Burgers equation,%
\footnote{The factor of 3 in the PDE form of the inviscid Burgers equation here signals the geometric notation to be used later for Lie transport by vector field $u^\sharp$ acting on a 1-form-density, e.g, $\mcal{L}_{u^\sharp} (m dx \otimes dx)\simeq (m u_x + (mu)_x)dx^2$ with $m=u$ and $u^\sharp$. \\In the KdV equation, though, the factor of 6 in the nonlinearity is traditional, following \cite{GGKM1967,Miura1976,ChevZHao2024}. The relevance of the ratio of coefficients of the nonlinearity and dispersion in the KdV equation for the qualitative result of collisions of the Burgers bore with KdV waves is demonstrated in Figure \ref{fig:gamma1_2} of Section \ref{sec: HB-KdV results}.}
\[ 
u_t + 3uu_x = 0
\,,\]
interacting with:
\begin{enumerate}
    \item 
Korteweg-de Vries (KdV) soliton solutions governed by 
\[
v_t + 6vv_x + \gamma v_{xxx} = 0  
\,.\]
As for the Burgers equation, the KdV equation is Galilean invariant in the sense that a given solution $v(x, t)$ remains a solution when `boosted' into a moving frame by replacing $x$ with $x + ct$ everywhere in $v(x, t)$ so that
\[
v (x,t)\mapsto v _{[c]}(x,t)=v (x+vt,t).
\]

\item 
The Nonlinear Schr\"odinger (NLS)
\[
i\hbar \psi_t = -\frac12 \psi_{xx} + \kappa|\psi|^2\psi\,,
\]
describes wave packets for a complex variable (wave function) $\psi(x,t)$. 
It has two types of solution known as focusing $(\kappa<0)$ and de-focusing $(\kappa>0)$.
The amplitude of the solution is given by $|\psi|^2=\psi^*\psi$.

The nonlinear Schr\"odinger equation is Galilean invariant in the sense that
a given solution $\psi(x, t)$ remains a solution when `boosted' into a moving frame by replacing $x$ with $x + ct$ everywhere in $\psi(x, t)$ while also multiplying by a phase factor of ${\displaystyle e^{-iv(x+ct/2)}\,}$ so that
\[
\psi (x,t)\mapsto \psi _{[c]}(x,t)=\psi (x+ct,t)\;e^{-ic(x+ct/2)}.
\]
\end{enumerate}
\begin{remark}
Although one might expect that the Burgers bores would either `snowplow' the waves ahead or run them over whenever it encounters them, the coupled nonlinear wave equations will tell a different story. In fact, as we shall discuss, the interactions of Burgers bores with KdV and NLS solutions can be much more profound than a simple `snowplow' effect. 
\end{remark}

\textbf{Approach.}
The equations governing wave-current dynamics in two dimensions have been derived and exemplified in \cite{HHS2023}. The equations governing the coupling dynamics exemplified here between the Burgers ramps-and-cliffs\footnote{Although the Euler-Poincar\'e variational derivation produces the inviscid Burgers equation, both global well-posedness of solutions and stability of the numerical simulations of the Burgers ramp-and-cliff dynamics require viscous regularisation.} and the KdV and NLS solitons will be derived following the work of \cite{HHS2023}. Namely, the \emph{Bore-Soliton} equations treated here will be derived by Hamilton's principle in a variational framework which couples the sum of two Lagrangians for the separate bore and soliton degrees of freedom via insertion of a vector field representing the Burgers current velocity into a 1-form density representing the momentum map for each type of soliton. A variant of this general approach was introduced by Dirac and Frenkel \cite{FrenkelDirac1934} in coupling the Schr\"odinger equation probability current density $J$ with the electromagnetic field vector potential $A$ to study linear nonrelativistic quantum electrodynamics (QED).  The quantum-classical $(J\cdot A)$ coupling of Schr\"odinger wave functions and Maxwell fields is known to produce profound cooperative effects, such as stimulated emission of radiation. The present work investigates what effects may occur when the Burgers current velocity is coupled to the dynamics of two well-known nonlinear wave soliton equations KdV and NLS  via their momentum maps arising in their corresponding phase-space variational principles.

\subsection{Examples}
\subsubsection{Burgers-KdV (B-KdV) dynamics}
For 1D wave-current interaction in the B-KdV case, the approach discussed here amounts to the composition of the two smooth invertible maps that govern the dynamics of the two continuum variables comprising the respective solutions for the bore momentum 1-form density ($u\,dx^2\in \Lambda^1(\mathbb{R})\otimes \mathrm{Den}(\mathbb{R})$) and the KdV soliton density ($v\,dx\in\mathrm{Den}(\mathbb{R})$). The Burgers velocity vector field $u^\sharp := u\partial_x\in\mathfrak{\mathbb{R}}$ is tangent to the right action of smooth invertible maps of the real line $\mathbb{R}$ onto itself by the diffeomorphism Lie group. The same map governs KdV dynamics, except it is augmented by the Gel'fand-Fuchs 2-cocycle, which introduces the third-order dispersion term in KdV and which together with the diffeomorphisms comprises the Bott-Virasoro Lie group. Thus, the interaction of these two types of \emph{coherent structures} for Burgers ramps-and-cliffs and KdV solitons will be governed by the composition of two time-dependent, smooth, Lie-group transformations of the real line $\mathbb{R}$ onto itself, in which one of these maps is extended by the Gel'fand-Fuchs 2-cocycle. 

As discussed below in Section \ref{sec: HB-KdV results}, Hamilton's principle for the sum of Lagrangians for Burgers and KdV equations coupled via the product of the fluid velocity and the wave momentum map yields the Burgers-KdV system of partial differential equations,%
\footnote{The Burgers-KdV \emph{system} in \eqref{eqns: HB+KdV_intro} is not in the same category as the Burgers-KdV \emph{equation} in \cite{Wazwaz2010}.}
\begin{align}
\begin{split}
\partial_t u + 3uu_x &= - v\, \p_x \big(\gamma v_{xx} + 3v^2 \big)
\,,
\\
\partial_t v  +\partial_x (uv) &= -\,\partial_x \big(\gamma v_{xx} + 3v^2 \big)
\,.
\end{split}
\label{eqns: HB+KdV_intro}
\end{align}
The right-hand sides of these equations arise represent the coupling between the two individual equations. The KdV dynamics for its potential velocity $v\,dx=\phi_x\,dx = d\phi$ can be regarded as being `swept' or `advected' as a density by the velocity vector field, $u^\sharp$, of the Burgers fluid current as,
\begin{align}
\begin{split}
(\partial_t  + \mathcal{L}_{u^\sharp})u\,dx^2 
&= -\, d \big(\gamma v_{xx} + 3v^2 \big)\otimes v\,dx
\,,
\\
(\partial_t  + \mathcal{L}_{u^\sharp})v\,dx 
&= -\,d \big(\gamma v_{xx} + 3v^2 \big)
\,.
\end{split}
\label{eqns: HB+KdV_geom}
\end{align}
However, this `advection' is not passive. Solutions of the B-KdV dynamics in \eqref{eqns: HB+KdV_intro} and their equivalent geometric form in \eqref{eqns: HB+KdV_geom} are observed in numerical simulations shown in the figure below to exert profound effects on the Burgers velocity vector field $u^\sharp$ that `sweeps' it. 

The figure below shows the typical Burgers-KdV overtaking collision.
\begin{figure}[H]
\begin{center}
\includegraphics[width=\textwidth]{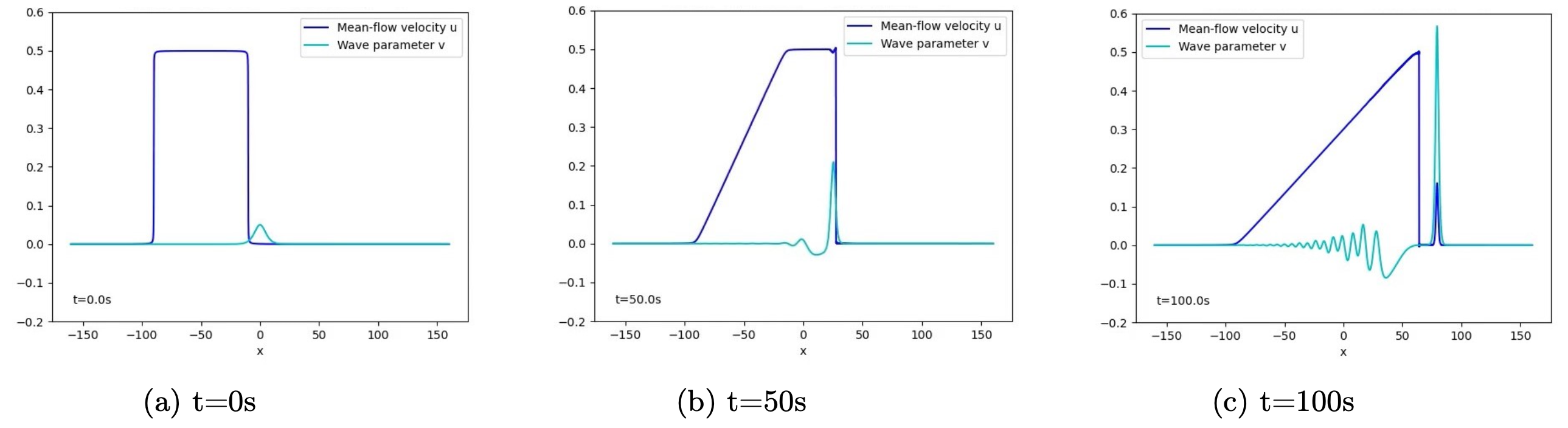}
\caption{The plots show the evolution of the mean-flow Burgers rightward velocity $u$ and the wave parameter $v$ in the coupled Burgers-KdV system of equations in 
\eqref{eqns: HB+KdV_intro} 
fromfor a small viscosity of $\nu = 0.01$ in the Burgers equation to stabilise the numerical simulation.  At time $t=0$, the Burgers bore overtakes the KdV soliton. At time $t = 50s$ the Burgers bore has started transferring momentum to the KdV wave and a wave moving rightward in the Burgers frame is developing. At time $t = 100s$ one sees the compound Burgers-KdV wave advancing ahead of the Burgers bore and leaving behind a leftward moving KdV wave train as viewed from the leading edge of the instantaneous Burgers motion.}
\end{center}
\end{figure}\vspace{-5mm}

\begin{figure}[H]
    \centering
    \includegraphics[width=\textwidth]{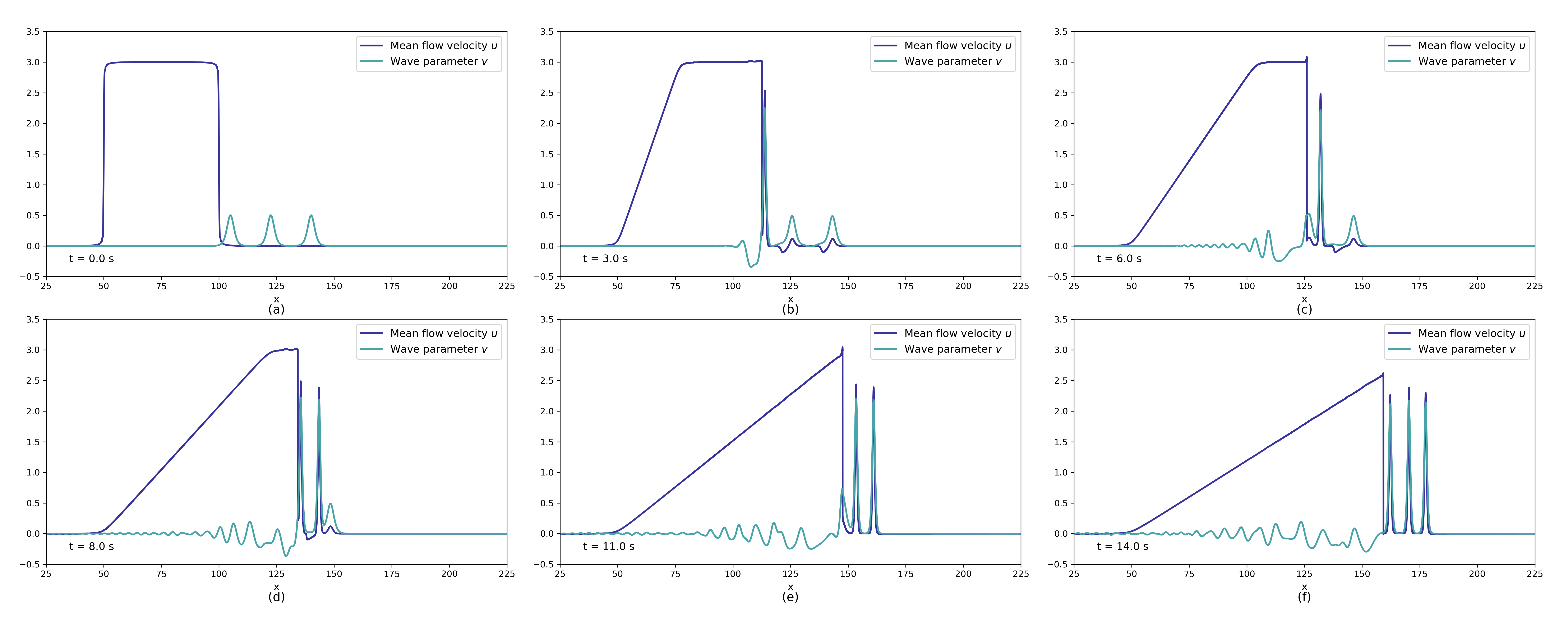}
    \caption{At time $t=0$, the Burgers bore overtakes the 1st of three identical rightward moving KdV waves in the Burgers frame. At time $t = 3s$, the Burgers bore has started transferring momentum to the first KdV wave, a KdV wave moving leftward in the frame of the bore is developing and the leading 2nd and 3rd KdV waves are beginning to create small Burgers waves. At time 6s, the 1st KdV wave has transferred most of its momentum to the 2nd (middle) KdV wave and a KdV wave train is moving leftward. At time 8s, momentum transfer from the bore has restored the amplitude of the 1st KdV wave and the 2nd KdV wave is overtaking the 3rd (rightmost) one. At time t=11s, the 2nd KdV wave has transferred its momentum to the rightmost 3rd wave and both of them have entrained part of the bore in becoming compound travelling waves. At t=14s, all three KdV waves have become rightward moving compound Burgers-KdV travelling waves. The middle wave will eventually overtake and transfer momentum to the leading wave, so that the heights of the compound waves will be ordered in velocity. }
    \label{fig:bore_three_soliton}
\end{figure}\vspace{-5mm}

\subsubsection{Burgers-NLS} 
The Burgers Ramp/Ciiff solution and the NLS solitons interact quite differently from the interactions of Burgers Ramp/Ciiff solutions and the KdV solitons. The Lie-Poisson form and the canonical Hamilton's canonical equations represented by the polar decomposition $\psi=\sqrt{N}\exp{i\phi}$ for the Burgers-NLS interaction are given by
\begin{equation}
    \begin{split}
&({\partial _t} + \mathcal{L}_{u^\sharp})( u - N\p_x \phi ) ( dx \otimes dx ) = 0
\,,\\
&({\partial _t} + \mathcal{L}_{u^\sharp+\phi_x^\sharp})(N\,dx)
\,,\\
&{\partial _t}\phi + u  \phi_x  =   - \frac{1}{2}  \phi_x^2 -\, \frac{ (\sqrt{N})_{xx} }{2 \sqrt{N}} + F^{\prime}(N)
\,.    \end{split}
\end{equation}
The figure below shows an example of the Burgers-NLS interaction equations derived in Section \ref{sec:NLS}.

\begin{figure}[H]
    \centering
    \includegraphics[width=\textwidth]{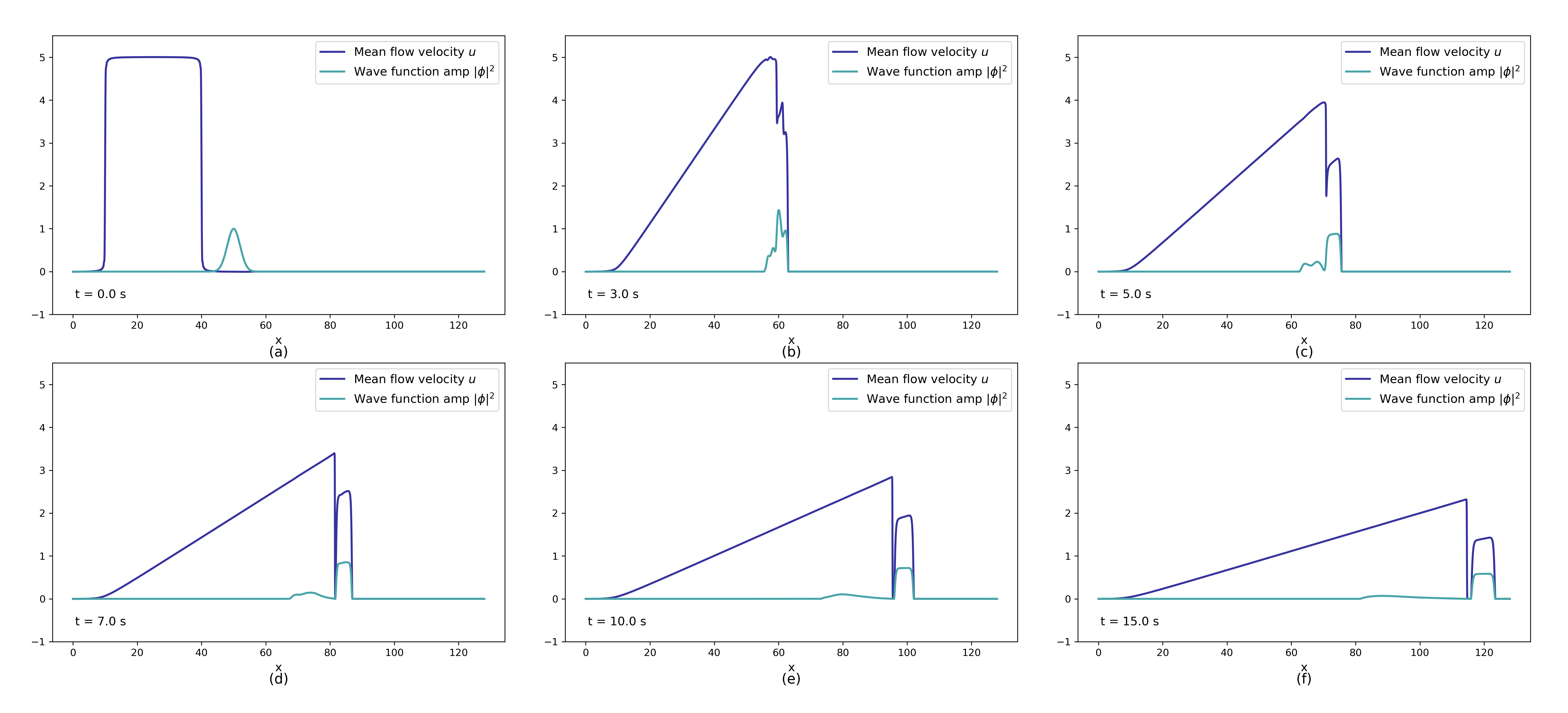}
    \caption{This plot illustrates the evolution of the mean-flow velocity $u$ and the wave function amplitude $|\phi|^2$ in the coupled Burgers-Nonlinear Sch\"odinger (Burgers-NLS) system of equations treated in section \ref{sec:NLS}.  The Burgers equation ramp/cliff solution is regularised by viscosity of $\nu_1=0.1$ in this case and dissipation $\nu_2=0.1$ has been added to the NLS equation to reduce the frequency of its phase oscillations. At time $t=0$, the Burgers bore starts to overtake the NLS Gaussian wave packet. At $t=4$, the bore is clearly transferring momentum to the NLS wave package. At $t=15$, a compound Burgers-NLS wave advances ahead of the ramp/cliff formation of the bore.}
    \label{fig:nls}
\end{figure}

\paragraph{Plan of the paper.}
\begin{itemize}
\item
Section \ref{EPDiff-SWwaves-sec} provides the background materials for shallow water waves and in one spatial dimensions. The examples of inviscid Burgers', Korteweg–De Vries (KdV) and Camassa Holm (CH) equations are discussed.  
\item
Section \ref{sec: HB-KdV results} discusses the Burgers-Korteweg-de Vries (B-KdV) results. In particular, Figure \ref{fig:gamma1_2} demonstrates the sensitivity of these results to the balance between nonlinearity and dispersion in the KdV nonlinear wave subsystem of the B-KdV collision. Additionally, we briefly consider the introduction of stochasticity into the B-KdV equations via the approach of Stochastic Advection by Lie Transport (SALT).

\item
Section \ref{sec:NLS} derives the one dimensional coupling of Burgers dynamics to the NLS equations. Note that this is a 1D analogue of the 2D coupling of Euler to NLS considered in \cite{HHS2023}.
\item
Section \ref{sec: Conclude} provides a brief summary of the results in this paper and an outlook for further developments.

%

\end{itemize}

\section{EPDiff and Shallow Water Waves}\label{EPDiff-SWwaves-sec}

\subsection{Introduction to wave equations}

Wave equations are evolutionary equations for time dependent curves in a space of smooth maps $C^\infty(\mathbb{R}^n,V)$ for solutions, $u\in V$, taking values in a vector space $V$. 
\begin{equation}
\partial_t u = f(u)
\,,\quad\hbox{or}\quad
\partial_t u_i(\mathbf{x},t) = f_i(u_i,\, u_{i,j},\, u_{i,jk},\, u_{i,jkl},\, \dots\, )
\,.
\label{eq: wave}
\end{equation}
Typically, $V$ is $\mathbb{R}$ or $\mathbb{C}$, $n=1$.
We are interested in the Cauchy problem. Namely, solve \eqref{eq: wave}  for $u(x,t)$, given the initial condition $u(x,0)$ and boundary conditions $u(x|_{\partial D},t)$.

\paragraph{Travelling waves.}
The simplest wave solution is called a travelling wave. This solution is a
function $u$ of the form 
\[
u(x, t) = F(x - ct)
\,,\] 
where $F : \mathbb{R} \to V$ is a function defining the wave
shape, and $c$ is a real number defining the propagation speed of the wave. 
Thus, travelling waves preserve their shape and simply translate to the right at a constant speed, $c$. 

\paragraph{Plane waves.}
A complex-valued travelling wave, called a plane wave,  plays a fundamental role in the theory of linear wave equations. 
The general form of a plane wave velocity is 
\[
u(x, t) = \Re e (A e^{i(kx-\omega t)}), 
\]
where $|A|$ is the wave amplitude, $k$ is the wave number, $\omega$ is the wave frequency, and $c_p=\omega/k$ is the speed along the oscillating waveform. 


\subsection{Conservation laws}

Conservation laws for evolutionary equations of the form $u_t = f(u) $ satisfy
\[
\frac{d}{dt}\int F(u)\,dx = \int \frac{\delta F}{\delta u} u_t\,dx  = \int \frac{\delta F}{\delta u} f(u) \,dx = \int dG(u) = 0
\,,
\]
for some functions $F$ and $G$ of $u$ and its derivatives, and for suitable boundary conditions.  

For example, the inviscid Burgers equation 
\begin{equation}
u_t + uu_x = 0 \,,
\label{eq: Burgers}
\end{equation}
has an infinite number of conservation laws, given by $C_n=\int \frac{u^n}{n}\,dx$
\begin{equation}
\frac{dC_n}{dt} =
\frac{d}{dt}\int \frac{u^n}{n}\,dx = \int u^{n-1}u_t\,dx = - \int u^nu_x\,dx 
= -\int \frac{1}{n+1}\partial_xu^{n+1}\,dx  = - \frac{1}{n+1}\int d(u^{n+1}) = 0
\,,
\label{eq: BurgersCLs}
\end{equation}
for homogeneous boundary conditions and any integer $n$. 


Even so, the solutions of the inviscid Burgers equation carry the seeds of their own destruction, since they exhibit wave breaking in finite time. That is, without dissipation or dispersion their velocity profile would develop a negative vertical slope in finite time. This is shown in the proof of the following Lemma.

\begin{lemma}[Steepening Lemma for the inviscid Burgers equation]\label{BurgersWaveBreakingLemma}\rm$\quad$\\
Suppose the initial profile of velocity $u(0,x)$ for the inviscid Burgers equation \eqref{eq: Burgers} has an inflection point of negative slope $u_x(0,\overline{x}(0))<0$ located at $x=\overline{x}(0)$ to the right of its maximum, and otherwise it decays to zero in each
direction sufficiently rapidly for all of its conservation laws in equation \eqref{eq: BurgersCLs}
to be finite. Then the negative slope at the inflection point will become vertical in
finite time.
\end{lemma}

\begin{proof}
Consider the evolution of the slope at the inflection point, defined by 
$s(t)=u_x(\overline{x}(t),t)$. Then the inviscid Burgers equation \eqref{eq: Burgers} yields an evolution equation for the slope, $s(t)$. Namely, using $u_{xx}(\overline{x}(t),t)=0$ the spatial derivative of equation \eqref{eq: Burgers} leads to 
\begin{equation} \label{slope-eqnBurgers}
\frac{ds}{dt} =-\, s^2
\quad\Longrightarrow\quad
s(t) = \frac{s(0)}{1+ s(0)t}
\,.
\end{equation}
Thus, if $s(0)<0$, the slope at the inflection point  $s(t)$ will become increasingly more negative, until it becomes vertical at time $t= -1/s(0)$. 
\end{proof}


\subsection{Survey of weakly nonlinear water wave equations: KdV and CH}

The derivation of weakly nonlinear water wave equations starts  
with Laplace's equation for the velocity potential of an inviscid,
incompressible, and irrotational fluid moving in a vertical plane under
gravity with an upper free surface, as, e.g., in \cite{Whitham[1974]}.

The equations are then expanded in the small parameters
$\epsilon_1 = a/h$ and $\epsilon_2 = h^2/l^2$. Here
$\epsilon_1 \ge \epsilon_2 > \epsilon_1^2$ and $a$, $h$, and $l$ denote the
wave amplitude, the mean water depth, and a typical horizontal
length scale (e.g., a wavelength), respectively. 
Length is measured in terms of $l$, height in $h$ and time in
$l/c_0$. The elevation $\eta$ is scaled with $a$ and fluid velocity $u$
is scaled with $c_0 a / h$. Here, $c_0=\sqrt{gh}$ is the linear wave speed 
for undisturbed water at rest at spatial infinity, where $u$ and its derivatives
$u_x$ and $u_{xx}$ are taken to vanish.

The result of the expansion to quadratic order in $\epsilon_1$ and $\epsilon_2$ 
is the equation for the surface elevation
$\eta$ (see e.g. \cite{Whitham[1974]}, p. 466), while higher order terms ($HOT$)
can be found in e.g. \cite{LiZhiSibgatullin[1997]},
\begin{eqnarray}
\label{asymtotic1}
0 = \eta_t+\eta_x+\frac{3}{2}\epsilon_1\,\eta\, \eta_x
+\frac{1}{6}\epsilon_2\,\eta_{xxx}
-\frac{3}{8}\epsilon_1^2\,\eta^2\,\eta_x
+\,\epsilon_1\epsilon_2\left(\frac{23}{24}\eta_x\,\eta_{xx}
+ \frac{5}{12}\eta\,\eta_{xxx}\right) + 
\epsilon_2^2 \frac{19}{360}\eta_{xxxxx} + \hbox{\small HOT}
\label{NLW-AsympExspan}
\end{eqnarray}
where partial derivatives are denoted by subscripts. 


\medskip

Next, following Kodama
\cite{Kodama[1985&1987],Kodama-Mikhailov[1996]}
one applies the near-identity transformation,
\[
\eta = u +
\epsilon_1\,f(u) +
\epsilon_2\,g(u)
\,,\]
to the $\eta-$equation (\ref{asymtotic1}) and seeks functionals $f(u)$ and $g(u)$ that consolidate the terms of order $O(\epsilon_1^2)$ and $O(\epsilon_1\epsilon_2)$ in \eqref{asymtotic1} into one order $O(\epsilon_2^2)$ term under normal form transformations. This procedure produces the following 1+1 quadratically nonlinear Camassa-Holm (CH) equation
for unidirectional water waves with fluid velocity, $u\,(x,t)$ and momentum $m=u-\alpha^2u_{xx}$, with constant $\alpha^{\,2} = (19/60)\epsilon_2$, see \cite{DGH[2001]},
\begin{equation}\label{CH-BV1}
m_t + c_0 u_x + \frac{\epsilon_1}{2}(um_x + 2mu_x)
+ \epsilon_2\frac{3}{20}u_{xxx} = 0
\,.
\end{equation}
After these normal form transformations, equation \eqref{CH-BV1} is equivalent to the shallow water wave
equation \eqref{asymtotic1} up to, and including, terms of order ${\cal O}(\epsilon_2^2)$.
For $\alpha^2$ positive, equation (\ref{CH-BV1}) becomes the Camassa-Holm equation derived and shown to be completely integrable in \cite{CH[1993]}. Hereafter, we will take $\alpha^2=0$ and leave the Burgers-CH interaction for later work.

For $\alpha^2=0$, equation (\ref{CH-BV1}) remains completely integrable, as it restricts to the Korteweg-de Vries (KdV) equation,
\begin{equation}\label{KdV-eqn1}
u_t + c_0 u_x + 3u\,u_x + \gamma\, u_{xxx} = 0
\,,
\end{equation}
which admits the soliton solution
$u(x,t)=u_0\,{\rm sech}^2((x-ct)\sqrt{u_0/\gamma}/2)$, $c = c_0+u_0$
see, e.g., \cite{Ablowitz&Segur[1981]}.

\begin{remark}[Interacting solutions at different orders]$\,$

The higher order terms in the asymptotic expansion of the nonlinear shallow water wave equations in equation \eqref{NLW-AsympExspan} represent degrees of freedom that are not accessed by the lower order terms.\footnote{This remark also applies to the asymptotic expansion of the corresponding Lagrangian in Hamilton's principle for the underlying fluid theory. See Gjaja and Holm \cite{GjajaHolm1996} for discussion of the further benefits of applying asymptotic expansions of Hamilton's principle in hierarchies of fluid dynamical approximations.} 

This observation in combination with the Galilean invariance of the Burgers equation and the KdV equation at the next order of the expansion then raises the following question:
What happens when a KdV solution is boosted into the time-dependent frame of motion of a Burgers solution? 
This is the question we address in the present work.
\end{remark}

\section{Burgers -- Korteweg-de Vries (KdV) collisions}\label{sec: HB-KdV results}

In 1D, we consider Hamilton's principle for the formulation of the Burgers-KdV equations where the Lagrangian is given by the sum of the kinetic energy of the Burgers' solution and the Whitham Lagrangian for the KdV solution as follows,
\begin{align}
\begin{split}
0 = \delta S= \delta \int_0^T \ell(u,\phi)\, dt\,,\qquad \ell(u,\phi) := \int_{\mathbb{R}}  \frac12 |u|^2 + \frac12\phi_x\left(\phi_t + u\phi_x\right) + \left( \phi_x^3 - \frac{\gamma}{2} \phi_{xx}^2 \right)\,dx\,.
\end{split}
\label{def-BKdV-action}
\end{align}
In Hamilton's principle \eqref{def-BKdV-action}, the variation in $u$ is constrained to have the form $\delta u = \p_t \xi - \ad_u \xi$ obtained from the Euler-Poincar\'e theory \cite{HMR1998}. The arbitrary variation $\xi$ and the variation $\delta \phi$ are assumed to be arbitrary and vanishing at endpoints $t=0$ and $t = T$.  
As we will see, invoking Hamilton's principle with this Lagrangian yields KdV dynamics in the frame of motion of the Burgers equation, and the KdV dynamics acts back directly on the Burgers dynamics. 

Computing the variations in $u$ and $\phi$ in \eqref{def-BKdV-action} yields
\begin{align}
\begin{split}
    0=\delta S &= \int_0^T \scp{ u + \frac{\phi_x^2}{2}}{\delta u} - \scp{\phi_{tx} + (u \phi_x)_x + \gamma \phi_{xxxx} + 3 \phi_x \phi_{xx}}{\delta \phi}\,dt \\
    & = \int_0^T \scp{ u + \frac{v^2}{2}}{\p_t \xi - \ad_u \xi} - \scp{\partial_t v + (u v)_x + \gamma v_{xxx} + 3 vv_x}{\delta \phi}\,dt \\
    & = \int_0^T \scp{ \left(\p_t + \ad^*_u \right)\left(u + \frac{v^2}{2}\right)}{\xi} - \scp{\partial_t v + (u v)_x + \gamma v_{xxx} + 3 vv_x}{\delta \phi}\,dt\,,
\end{split}
\label{var-u-phi}
\end{align}
where in the second line we have inserted the constrained variations of $u$ and introduced the one-form $v \,dx= \phi_x \,dx$. In equation \eqref{var-u-phi}, the angle bracket operation $\scp{\,m}{\xi\,}: (\Lambda^1(\mathbb{R})\otimes {\rm Den}(\mathbb{R}) )\times \mathfrak{X}(\mathbb{R})\to \mathbb{R}$
denotes the $L^2$ dual pairing
\begin{align}
\scp{\,m}{\xi\,}:= \int_{\mathbb{R}}  \xi m \,dx
:= \int_{\mathbb{R}}  (\xi\p_x)\contract  (mdx)dx
\,,
\label{def-pair-brkt}
\end{align}
where $(\xi\p_x)\contract  (mdx)dx$ denotes insertion of a vector field $\xi\p_x$ into a differential 1-form density $m\,dx\otimes dx$, which one may be abbreviated as $m\,dx^2$ without confusion.\footnote{For a review of differential form notation and usage in fluid dynamics see \cite{Holm2011}.}
The arbitrary variation in $\phi$ yields the equation
\begin{align}
\partial_t v + \partial_x(uv  + \gamma v_{xx} + 3v^2) = 0 \,,\label{eqn-BKdV-vEgn}
\end{align}
where we see that $v$, the solution that corresponds to the KdV part of the flow, is swept along by the HB $u$-solution. That is, the one-form $v\,dx$ is Lie transported by the vector field $u\p_x$. In terms of the Lie derivative ${\cal L}_{u^\sharp}$, one can write the $v$-equation \eqref{eqn-BKdV-vEgn} in coordinates on the real line as, 
\begin{align}
\big(\partial_t  + {\cal L}_{u^\sharp} \big)(v\,dx) = \big(\partial_t v + \partial_x(uv)\big)dx 
= -\, d( \gamma v_{xx} + 3v^2)
\,,\label{LieXport-v}
\end{align}
where the exterior derivative $d$ represents the spatial differential.

The arbitrary variations in $\xi$ yields dynamics for the total momentum 1-form density $m$, defined by 
\begin{align}
m:= {\delta \ell}/{\delta u} = u + \frac12 v^2\,.
\label{Total-BKdV-mom}
\end{align}
Noting that $\ad^*_{u^\sharp} m = \mcal{L}_{u^\sharp} m$ when $m$ is a one-form density, the dynamics of $m$ can be written as 
\begin{align}
\big(\partial_t  + {\cal L}_{u^\sharp} \big)(m\,dx^2) = \big(m_t + (\partial_x m + m \partial_x)u \big)dx^2  = 0 \,,\quad\hbox{with}\quad m := u+\tfrac12v^2 \,,
\label{eqn-BKdV-mEgn}
\end{align} 
where we have used the coordinate expression of $\mcal{L}_{u^\sharp}$ on one-form densities. Thus, we may collect the Lie derivative forms of the Burgers-KdV equations as
\begin{align}
\big(\partial_t  + {\cal L}_{u^\sharp} \big)(m\,dx^2) = 0
\quad\hbox{and}\quad
\big(\partial_t  + {\cal L}_{u^\sharp} \big)(v\,dx) = -\, d\big(\gamma  v_{xx} + 3 v^2\big)
\,.\label{eqn-myEqns-Lie}
\end{align} 

A short calculation to eliminate $m=u+\tfrac12v^2$ in favour of $u$ in \eqref{eqn-BKdV-mEgn} using the $v$-equation in \eqref{eqn-BKdV-vEgn} finally shows that the results of Hamilton's principle in equation \eqref{var-u-phi} yields the system in equation \eqref{eqns: HB+KdV_intro}.

Hence, the velocity equation and
the momentum density equation in \eqref{eqn-myEqns-Lie} together imply 
via the product rule for the Lie derivative that
\begin{equation}\begin{split}
\left( \partial _t + {\cal L}_{u^\sharp} \right)(u\,dx^2) 
&=
-\,\frac{1}{2}\left( \partial _t + {\cal L}_{u^\sharp} \right)(v^2\,dx^2)
= - \, (vdx)\otimes ( \partial _t + {\cal L}_{u^\sharp} )(v\,dx)
\,,\\
\hbox{so that\quad}
(u_t + 3uu_x)dx^2 
&=
(vdx)\otimes d( v_{xx} + 3 v^2 )
=
\big(v\,\partial_x ( \gamma v_{xx} + 3 v^2)\big)\,dx^2
\,.\end{split}
\label{u-eqn}
\end{equation}
Thus, equations \eqref{eqn-myEqns-Lie} provide the geometric forms of the  
$(u,v)$ mutual interaction wave-current system in \eqref{eqn-myEqns-Lie}.
\begin{align}
\begin{split}
\partial_t u + 3uu_x &= - \,v\, \partial_x\big(\gamma v_{xx} + 3 v^2 \big)
\,,
\\
\partial_t v + \partial_x (uv) &=  -\,\partial_x\big(\gamma v_{xx} + 3 v^2 \big) 
\,.
\end{split}
\label{eqns: EPDiff+KdV1-Ham}
\end{align}
\begin{remark}
Homogeneous boundary conditions have been enforced for all spatial integrations by parts in the previous proof of the Euler-Poincar\'e equations arising from the Lagrangian in equation \eqref{def-BKdV-action}.
The definition $v\,dx=d\phi$ implies that the quantity $\overline{v} = \int_{\mathbb{R}} v(x,t) \,dx$ is constant in time for vanishing boundary conditions as $|x|\to\infty$. 
\end{remark}

\begin{remark}[Hamiltonian formulation for the Burgers-KdV interaction]
Upon writing the KdV velocity as $v:=\phi_x$ and using \eqref{Total-BKdV-mom} to write the Burgers kinetic energy in terms of $m$ and $v$, the natural Hamiltonian for the combined system of KdV `waves' interacting dynamically with the Burgers `current' may be taken as
\begin{align}
h(m,v) =   \int_{\mathbb{R}} \frac12\big( m - \tfrac12v^2  \big)^2 + \frac{\gamma}{2} v_x^2 - v^3    \,dx
\,.\label{eqn: h(m,v)}
\end{align}
The corresponding variational derivatives are given by
\[
\delta h(m,v) =  \int_{\mathbb{R}}  u\,\delta m - \big( uv + 3v^2 
+ \gamma v_{xx} \big)\,\delta v\,dx
\,.\]
Consequently, one may express Hamilton's equations for the Burgers-KdV equations as 
\begin{align}
\begin{split}
\partial_t m&= - \big( \partial_x m + m  \partial_x \big) \frac{\delta h}{\delta m}
= - \big( \partial_x m + m  \partial_x \big) u 
\,,
\\
\partial_t v &= \partial_x  \frac{\delta h}{\delta v} = - \,\partial_x (uv  + \gamma v_{xx} + 3v^2) 
\,.
\end{split}
\label{eqns: EPDiff+KdV1}
\end{align}
The Hamiltonian equations for the Burgers-KdV dynamics in \eqref{eqns: EPDiff+KdV1} may also be written in diagonal matrix Poisson operator form, as%
\begin{align}
 \partial_t
 \begin{bmatrix}
      m
      \\
      v   
 \end{bmatrix}
 = -
 \begin{bmatrix}
     \partial_x m + m\partial_x  & 0 \\
     0  & \partial_x 
 \end{bmatrix}
 \begin{bmatrix}
      \fder{h}{m} = u \\
      \fder{h}{v} = (uv +\gamma v_{xx} + 3 v^2)
 \end{bmatrix}
 .\label{eqn-with3}
\end{align}
Here we see that the  Hamiltonian structure of the $m$-equation in \eqref{eqn-with3} is Lie-Poisson, as expected from the Euler-Poincar\'e reduction. We also see the \emph{second} Hamiltonian structure for the KdV equation, whose  Poisson operator is simply the spatial partial derive $\p_x$. This Hamiltonian structure does not involve introducing the Bott-Virasoro Lie algebra, as needed for the other Hamiltonian structure with the Hamiltonian $\tfrac12\int_{\mathbb{R}}v^2dx$ in order to capture the dispersive term $\gamma v_{xxx}$ in the KdV equation \cite{Marsden[1999]}. 
\end{remark}

\begin{remark}[Casimirs for the diagonal (untangled) Poisson operator in \eqref{eqns: EPDiff+KdV1}]
Casimirs are functionals that Poisson commute with any other functional of the Hamiltonian variables, in this case $(m,v)$. In general, the variational derivative of a Casimir functional is a null eigenvector of the Poisson operator. In the present case, the Casimirs are 
\[
C_m =  \int_{\mathbb{R}}  \sqrt{m}\,dx
\quad\hbox{and}\quad
C_v =  \int_{\mathbb{R}}  v\,dx
\,.\]

In addition, for the present case, the Poisson brackets among the moments $f_m(v)=\int_\mathbb{R} v^m \,dx$ commute among themselves,
\[
\{f_m,f_n\} = 0
\,.\]
\end{remark}

\begin{remark}[The tangled Poisson operator \eqref{eqns: EPDiff+KdV1}]
Transforming variables in the Hamiltonian in \eqref{eqn: h(m,v)} from $h(m,v)$ to $h(u,v)$, by substituting  
\begin{align}
u = m - \tfrac12v^2 
\,,\label{eqn-FibreDeriv}
\end{align}
leads to the equivalent Hamiltonian,  
\begin{align}
h(u,v)=\int \frac12 u^2 + \frac12 v_x^2 - v^3  \, dx
\,.\label{eqn-TotHam}
\end{align}
The corresponding equivalent equations in Hamiltonian form in the transformed variables $(u,v)$ may be expressed in terms of the semidirect product Lie-Poisson equations with a generalised 2-cocycle as
\begin{align}
 \partial_t
 \begin{bmatrix}
      u\\
      v   
 \end{bmatrix}
 = -
 \begin{bmatrix}
     u\partial_x + \partial_x u & v\partial_x\,   \\
     \partial_x v  & \partial_x 
 \end{bmatrix}
 \begin{bmatrix}
      \fder{h}{u} = u\\
      \fder{h}{v} = \gamma v_{xx} + 3 v^2
 \end{bmatrix}
 =
 - \begin{bmatrix}
      3\,uu_x + v \partial_x (\gamma v_{xx} + 3 v^2) \\
      \partial_x (uv) + \partial_x ( \gamma v_{xx} + 3 v^2) 
 \end{bmatrix}
.\label{eqn-with4}
\end{align}
Evolution under the Lie-Poisson bracket corresponding to the Poisson operator in equations \eqref{eqn-with4} is given by
\begin{align}
\frac{df}{dt} = \{f,h\} (u,v) = - \int_\mathbb{R} 
\begin{bmatrix}
\delta{f}/\delta{u} \\  \delta{f}/\delta{v}
\end{bmatrix}^T
\begin{bmatrix}
     u\partial_x + \partial_x u & v\partial_x\,   \\ 
     \partial_x v & \partial_x 
 \end{bmatrix}
\begin{bmatrix}
\delta{h}/\delta{u} \\  \delta{h}/\delta{v}
\end{bmatrix}dx
\,.\label{eqn-LPB2cocycle}
\end{align}
The Poisson bracket in equation \eqref{eqn-LPB2cocycle} is the sum of a Lie-Poisson bracket dual to the semidirect-product Lie algebra $\mathfrak{X}(\mathbb{R})\circledS {\rm Den}(\mathbb{R})$ of vector fields $\mathfrak{X}(\mathbb{R})$ acting on densities ${\rm Den}(\mathbb{R})$ on the real line $\mathbb{R}$   with dual coordinates $u\in \mathfrak{X}^*(\mathbb{R})$ and $v\in {\rm Den}^*(\mathbb{R})$ plus constant antisymmetric bracket with  $\partial_x$ in the $\{v,v\}$ position, inherited as a central extension from the bi-Hamiltonian structure of the KdV equation. By skew symmetry of the Poisson operator in this bracket operation under the $L^2$ pairing, the equations in \eqref{eqn-with4} preserve the  Hamiltonian $h(u,v)$ in equation \eqref{eqn-TotHam} above, since of course \{h,h\} vanishes identically.
\end{remark}
\begin{figure}[H]
    \centering
    \includegraphics[width=\linewidth]{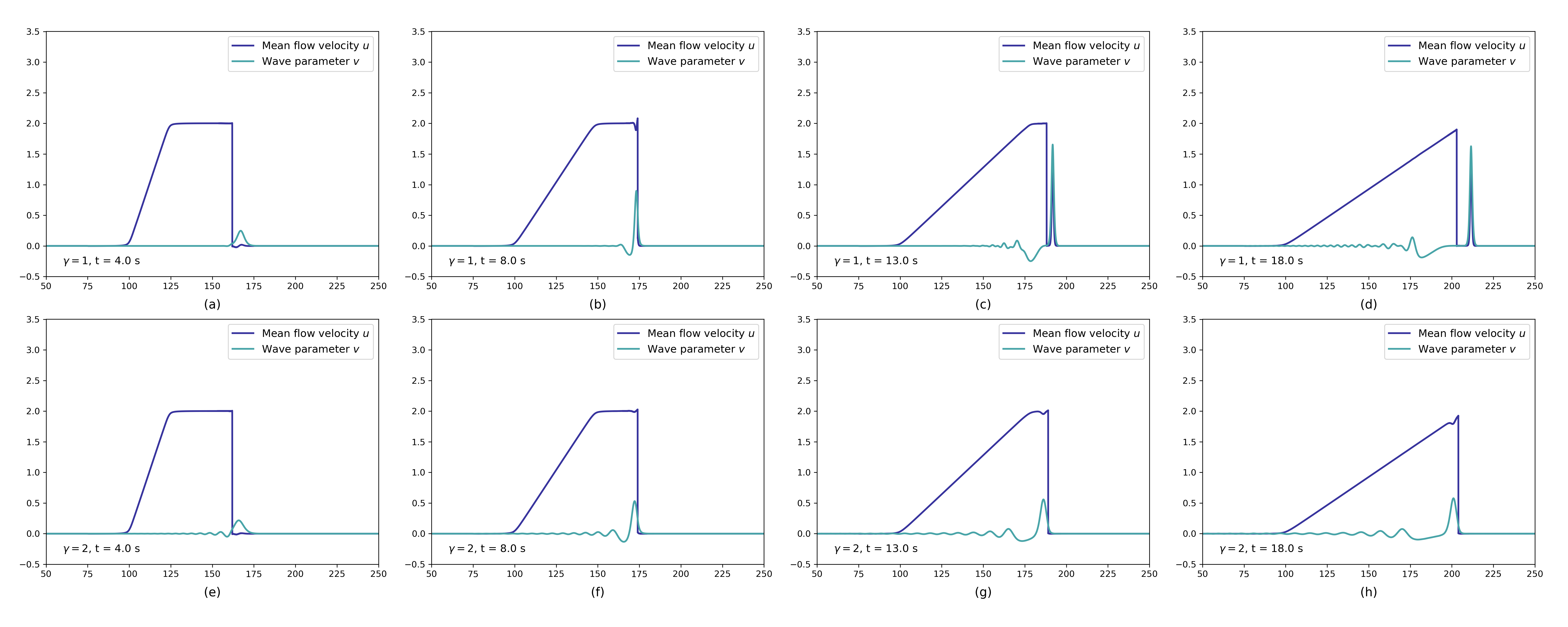}
    \caption{These figures represents the wave-current interaction with different dispersion strengths $\gamma =1$, and $\gamma =2$ at a fixed value of 6 for the the KdV nonlinearity coefficient. Fig (a)-(d) shows that for $\gamma=1$ a compound wave is excited and propagates ahead away from the bore. In contrast, Fig (e)-(h) shows that for $\gamma=2$, the KdV wave moves along with the bore front. }
    \label{fig:gamma1_2}
\end{figure}
\paragraph{Numerical method} In the numerical study, we use the pseudo-spectrum method supported by Dedalus Project \cite{dedalus}. The computational domain is discretized into 32,768 points over a length of 256 units, utilizing a $3/2$ dealiasing factor to ensure accuracy in the Fourier spectral representation. We use a semi-implicit backward differentiation formula (SBDF4) scheme with a fixed timestep of $10^{-6}$.

\subsection{Introducing stochasticity into Burgers-KdV via the SALT approach}\label{sec: SALT}

Ideal fluid dynamics in the Eulerian representation admits a Lie symmetry reduced Euler-Poincar\'e formulation \cite{HMR1998}. In turn, the Euler-Poincar\'e formulation defines advective transport in ideal Eulerian fluid dynamics in terms of the Lie derivative operation. The Stochastic Advection by Lie Transport (SALT) approach introduces stochasticity into the Euler-Poincar\'e formulation of ideal fluid dynamics as a semimartingale for the transport velocity, 
\begin{align}
\rmd x_t = u(x,t)\mathrm{d}t + \sum_{i=1}^N \xi_i(x)\circ \mathrm{d}W^i_t
\,,\label{eqn-SALT}
\end{align}
thereby preserving the geometric form of the deterministic equations. In fact, the drift velocity $u(x,t)$ is the deterministic fluid transport velocity and the $\xi$'s are determined from principle component analysis of observed or simulated data, as discussed in \cite{Cotter-etal-2019,Cotter-etal-2020}.

Since we have formulated the Burgers-KdV equations \eqref{eqn-myEqns-Lie} in terms of deterministic Lie transport, we may reformulate them in SALT form as
\begin{align}
\big(\rmd  + {\cal L}_{\rmd x_t} \big)(m\,dx^2) = 0
\quad\hbox{and}\quad
\big(\rmd  + {\cal L}_{\rmd x_t} \big)(v\,dx) 
= -\, d( \gamma v_{xx} + 3v^2)\mathrm{d}t
\,,\quad\hbox{with}\quad
m := u+\tfrac12v^2
\,.\label{eqn-myEqns-Lie-SALT}
\end{align} 
An analysis of these SALT HB-KdV equations is deferred and will be discussed elsewhere.

\begin{remark}[B-KdV versus B-CH collisions]$\,$
The equations for B-CH collisions have the same geometric structure as for B-KdV,
and the Poisson operator in equation \eqref{BCH-LPform} contains a new `hydrodynamic form' for CH. In addition, if one transforms variables from density $v\,dx$ to 1-form density $\mu=(1-\p_x^2)v\,dx^2$, then the Poisson operation in \eqref{BCH-LPform} transforms to
\begin{align}
\p_t \begin{bmatrix} u \\ \mu \end{bmatrix}
=
- 
\begin{bmatrix} 
\ad^*_\Box u & \ad^*_\Box \mu
\\ 
\ad^*_\Box \mu & -\,\ad_{\,\p_x}\Box
\end{bmatrix}
\begin{bmatrix} \delta H/ \delta u \\ \delta H/ \delta \mu 
\end{bmatrix}
\,.\label{BCH-LPform-mu}
\end{align}
\end{remark}`
\todo[inline]{DH: Denoting $w=\delta H/ \delta \mu\in \mathfrak{X}$ we have $-\ad_{\,\p_x}w = w_x\p_x + w\p_x - w\p_x = \,w_x\p_x $ }

\end{comment}

\section{Burgers -- Nonlinear Schr\"odinger (B-NLS) collisions }\label{sec:NLS}


Consider the following application of Hamilton's principle for the Burgers-NLS dynamics, with $F(N) = \kappa N^2$,
\begin{equation}
\begin{split}
& 0 = \delta S=\delta \int_0^T \int_{\mathbb{R}} \frac{u^2}{2}- N\left(\phi_t+u \phi_x \right)
- \frac{N}{2} \phi_x^2 - \frac{1}{2}(\p_x \sqrt{N})^2+F(N) \,d x dt \,,\\
& =\int_0^T \int_{\mathbb{R}}(u- N \phi_x )\delta u
+\Big(-\big(\phi_t+u  \phi_x \big)
- \frac{1}{2}\phi_x^2 
+ \frac{ (\sqrt{N})_{xx} }{2 \sqrt{N}}+F^{\prime}(N)\Big) \delta N \\
&\qquad\quad +\big(N_t+ \p_x  \big(N(u + \phi_x )\big) \,\delta \phi \,d x dt
\,.\label{HamPrincEqns}
\end{split}
\end{equation}
To check these equations, consider the NLS Hamiltonian in the fixed \emph{laboratory} frame:
\begin{equation}
    H_{Lab}(N,\phi ) = \int \frac{N}{2} \phi_x^2 + \frac{1}{2}\big(\p_x \sqrt{N})^2  - F(N)\big)\,d x
\,,\end{equation}
with variations
\begin{equation}
    \delta H_{Lab} = \int {\left(  \frac{1}{2}\phi_x^2  + \frac{ (\sqrt{N})_{xx} }{2 \sqrt{N}} - F^{\prime}(N) \right) \delta N 
+ \left( { - \p_x (N \p_x \phi )} \right)\delta \phi \,d x} \,.
\end{equation}
In the laboratory frame one then has Hamilton's equations
\begin{equation}
    \begin{split}
{\partial _t}\phi  &= -\,\frac{{\delta H_{Lab}}}{{\delta N}} 
=  -\,\frac{1}{2}\phi_x^2 - \, \frac{ (\sqrt{N})_{xx} }{2 \sqrt{N}} +  F^{\prime}(N)\,,
\\
{\partial _t}N &=  \frac{{\delta H_{Lab}}}{{\delta \phi }} 
=  -\, \p_x  (N\phi_x )\,.
\end{split}
\end{equation}
In the Burgers frame of motion, we have $\delta u = {\partial _t}\xi  - \text{ad}{_u}\xi $ and the 
Hamiltonian  $H(N,\phi )$ is boosted into the Burgers frame by adding the momentum map coupling term
\[
H(N,\phi ) = H_{Lab}(N,\phi ) + \int \mathrm{u}  N \phi_x \,dx
\,.\]
The Lie-Poisson form and the canonical Hamilton's canonical equations
in the Burgers frame then become
\begin{equation}
    \begin{split}
&({\partial _t} + {\mathcal{L}_{u^\sharp})\Big(( u - N\p_x \phi )} ( dx \otimes dx )\Big) = 0
\,,\\
&({\partial _t} + {\mathcal{L}_{(u +\phi_x )^\sharp})(N\,dx)}  =  0
\,,\\
&({\partial _t} + {\mathcal{L}_{u^\sharp})\phi}  =   - \frac{1}{2}  \phi_x^2 -\, \frac{ (\sqrt{N})_{xx} }{2 \sqrt{N}} + F^{\prime}(N)
\,.    \end{split}
\end{equation}
These boosted canonical equations in geometric form reveal the transport operations in the B-NLS equations and they agree with the results of Hamilton's principle in equation \eqref{HamPrincEqns} when their coefficients are collected. In particular, though, they reveal where stochastic transport may be properly added in the B-NLS system. Namely, the addition of stochastic transport to the vector fields $u^\sharp$ or $\phi_x^\sharp$ may be added separately or together, provided the stochastic transport vector fields are uncorrelated.

\section{Conclusion and Outlook}\label{sec: Conclude}

This paper has derived and simulated 1D self-consistent models of wave-current interaction equations modelled by Burgers motion transporting KdV nonlinear wave evolution special initial conditions modelling the overtaking collisions of Burgers bores with KdV and NLS waves. In each case, we have stressed the generality of the derivations of the wave-current interaction equations via the composition-of-maps variational approach by writing the wave-current collision equations in coordinate-free differential form to reveal their geometric structure. We have also simulated the B-KdV and B-NLS equations computationally in order to illustrate their fascinating solution behaviour. 

B-NLS wave-current collisions can be generalised to higher dimensions. In fact, the composition-of-maps approach for the coupling of ideal Euler fluid dynamics to NLS waves via composition of maps in 2D has already been derived, investigated, simulated and discussed in detail in \cite{HHS2023}. However, the complexity of the 2D interactions of fluid flow with NLS waves seen in \cite{HHS2023} warrants 1D investigation to better illustrate the rich solution behaviour in a simpler context.

The coupled wave-current models studied here were also made stochastic using the SALT approach which preserves the variational derivations. The effects of stochasticity in other Burgers -- nonlinear wave interactions also deserve further investigation. 

\section*{Acknowledgements}
We are grateful to our friends, colleagues and collaborators for their advice and encouragement in the matters treated in this paper. 
We especially thank C. Cotter, D. Crisan, E. Luesink, for many insightful discussions of corresponding results similar to the ones derived here for other wave-current interactions. DH and OS are also grateful for partial support during the present work by European Research Council (ERC) Synergy grant DLV-856408 (Stochastic Transport in Upper Ocean Dynamics -- STUOD). 
The work of RH was supported by Office of Naval Research (ONR) grant N00014-22-1-2082 (Stochastic Parameterisation of Ocean Turbulence -- SPOT), and the work of HW was supported by US AFOSR Grant FA9550-19-1-7043 - FDGRP (Fluid Dynamics of Geometric Rough Paths -- FRGRP).


\begin{thebibliography}{99}


\bibitem{Ablowitz&Segur[1981]} M.J. Ablowitz and H. Segur, {\it
Solitons and the Inverse Scattering Transform}, SIAM: Philadelphia
(1981).

\bibitem{Arnold1966}
Arnold, V. I. (1966). Sur la g\'eom\'etrie diff\'erentielle des groupes de Lie de dimension infinie et ses
applications \'a l'hydrodynamique des fluides parfaits. Annales de l'institut Fourier.

\bibitem{Bloch-etal1996}
Bloch, A., Krishnaprasad, P. S., Marsden, J. E., \& Ratiu, T. S. (1996). The Euler-Poincar\'e equations
and double bracket dissipation (Vol. 175). Springer-Verlag.

\bibitem{dedalus} 
{Burns}, Keaton J. and {Vasil}, Geoffrey M. and {Oishi}, Jeffrey S. and {Lecoanet}, Daniel and {Brown}, Benjamin P.
Dedalus: A flexible framework for numerical simulations with spectral methods
{\it Phys. Rev. Research 2, 023068},
\url{https://doi.org/10.1103/PhysRevResearch.2.023068}

\bibitem{CH[1993]} R. Camassa and D.D. Holm, Phys. Rev. Lett. {\bf71},
1661 (1993).\\
\url{https://doi.org/10.1103/PhysRevLett.71.1661}

\bibitem{CeHoMaRa1999}
Cendra, H., D. D. Holm, J. E. Marsden and T. S. Ratiu [1999],
Lagrangian Reduction, the Euler--Poincar\'e
Equations, and Semidirect Products.
{\it Arnol'd Festschrift Volume II},
{\bf186}, Amer. Math. Soc. Translations Series 2, pp. 1--25.

\bibitem{CeMaRa2001}
Cendra, H., J. E. Marsden, and T. Ratiu [2001] 
Lagrangian Reduction by Stages. 
Mem. Amer. Math. Soc. {\bf 152}, no. 722, viii+108 pp.

\bibitem{ChevZHao2024}
Cheviakov, A. and Zhao, P., 2024. Shallow Water Models and Their Analytical Properties. In Analytical Properties of Nonlinear Partial Differential Equations: with Applications to Shallow Water Models (pp. 79-267). Cham: Springer International Publishing.

\bibitem{CotterBokhove2009}
Cotter, C. and Bokhove, O., 2010. Variational water-wave model with accurate dispersion and vertical vorticity. Journal of engineering mathematics, 67, pp.33-54.

\bibitem{Cotter-etal-2019}
Cotter, C.J.,  Crisan, D., Holm, D.D.,  Pan, W. and Shevchenko, I., 2019.
Numerically Modelling Stochastic Lie Transport in Fluid Dynamics,
SIAM Multiscale Model. Simul., 17(1), 192--232.\\
\url{https://doi.org/10.1137/18M1167929}



\bibitem{Cotter-etal-2020}
Cotter, C.J.,  Crisan, D., Holm, D.D.,  Pan, W. and Shevchenko, I., 2020.
Data Assimilation for a Quasi-Geostrophic Model with Circulation-Preserving Stochastic Transport Noise. 
\\J Stat Phys 179, 1186-1221. 
\url{https://doi.org/10.1007/s10955-020-02524-0}

\bibitem{DGH[2001]}
Dullin, H.R., Gottwald, G. and Holm, D.D. (2004).
On asymptotically equivalent shallow water wave equations.
Physica D 190, 1-14. \url{https://doi.org/10.1016/j.physd.2003.11.004}

\bibitem{EP-footnote}
Equation (\ref{CH-BV1}) also arises as the Euler-Poincar\'e
equation  for an averaged Lagrangian \cite{HMR-1998}.
Its solutions describe geodesic motion with respect to the $H_1$ metric of
$u$ on the Bott-Virasoro Lie group \cite{Misiolek[1998]}. The KdV equation
arises the same way for the $L_2$ metric, see \cite{MR[1999]}.

\bibitem{FrenkelDirac1934}
Frenkel, J. and Dirac, P.A.M., 1934. \textit{Wave mechanics: advanced general theory}. Clarendon Press Oxford.

\bibitem{GGKM1967}
Gardner, C.S., Greene, J.M., Kruskal, M.D. and Miura, R.M., 1967. Method for solving the Korteweg-deVries equation. Physical review letters, 19(19), p.1095.

\bibitem{GjajaHolm1996}
Gjaja, I. and Holm, D.D., 1996.
Self-consistent wave-mean flow interaction
dynamics and its Hamiltonian formulation for a rotating
stratified incompressible fluid.  
Physica D, 98, 343-378.
\url{https://doi.org/10.1016/0167-2789(96)00104-2} 

\bibitem{Holm2002}
Holm, D. D. (2002). Euler-Poincar\'e dynamics of perfect complex fluids. \emph{Geometry, Mechanics,
and Dynamics}, 169-180. Retrieved from \url{https://doi.org/10.1007/0-387-21791-6_4}

\bibitem{Holm2011}
Holm, D. D. (2011). \textit{Geometric Mechanics, Part I}. World-Scientific.

\bibitem{Holm2019}
Holm, D. D. (2019). Stochastic closures for wave-current interaction dynamics. Journal of
Nonlinear Science, 29 , 2987-3031. \url{https://doi.org/10.1007/s00332-019-09565-0}

\bibitem{HHS2022a}
Holm, D. D., Hu, R., \& Street, O. D. (2022a). Coupling of waves to sea surface currents via horizontal density gradients. Retrieved from \url{http://arxiv.org/abs/2202.04446}

\bibitem{HHS2023}
Holm, D. D., Hu, R., \& Street, O. D. (2022b, 12). Lagrangian reduction and wave mean flow
interaction. Physica D 454, Article 133847.
Retrieved from \url{https://doi.org/10.1016/j.physd.2023.133847}

\bibitem{HoKu1988}
Holm, D. D. and B. A. Kupershmidt [1988],
The analogy between spin glasses and Yang--Mills fluids,
{\it J. Math. Phys.} {\bf29}, 21--30.

\bibitem{HMR-1998}
Holm, D.D., Marsden, J.E. and Ratiu, T.S., 1998. 
Euler--Poincar\'e models of ideal fluids with nonlinear dispersion,
Phys. Rev. Lett., {\bf 80}, 41734177.
\url{https://doi.org/10.1103/PhysRevLett.80.4173}

\bibitem{HMR1998}
Holm, D.D., Marsden, J.E. and Ratiu, T.S., 1998. 
The Euler-Poincar\'e equations and semidirect products with applications to continuum theories. Advances in Mathematics, 137(1), pp.1-81. \url{https://doi.org/10.1006/aima.1998.1721}

\bibitem{HSS2009}
Holm, D. D., Schmah, T., \& Stoica, C. (2009). \emph{Geometric mechanics and symmetry: From finite
to infinite dimensions}. Oxford University Publishers.

\bibitem{HTY2009}
Holm, D., Trouv\'e, A., \& Younes, L. (2009). The Euler-Poincar\'e theory of metamorphosis. 
Quarterly of Applied Mathematics, 67 (4), 661-685. 
\url{https://doi.org/10.1090/S0033-569X-09-01134-2}

\bibitem{Kodama[1985&1987]}
Y. Kodama,
Phys. Lett. A {\bf 107}, 245, {\bf 112}, 193 (1985);
{\bf 123}, 276 (1987).

\bibitem{Kodama-Mikhailov[1996]}
Y. Kodama and A. V. Mikhailov, in {\it Algebraic Aspects of Integrable
Systems: In Memory of Irene Dorfman}, edited by A. S. Fokas and I. M.
Gelfand, Birkh\"auser, Boston, (1996) pp 173-204.

\bibitem{Marsden[1999]}
Marsden, J.E., 1999. Park City lectures on mechanics, dynamics, and symmetry. Symplectic Geometry and Topology, 7, pp.335-430.

\bibitem{MR[1999]}
J.E. Marsden and T.S. Ratiu,
{\it Introduction to Mechanics and Symmetry}, 2nd Edition,
   Springer:New York (1999).

\bibitem{Marsden et al.(2000)}
Marsden, J. E., T. S. Ratiu and J. Scheurle [2000],
Reduction theory and the Lagrange--Routh equations, 
J. Math. Phys., 41, 3379--3429.

\bibitem{Misiolek[1998]} 
Misiołek, G., 1998. A shallow water equation as a geodesic flow on the Bott-Virasoro group. Journal of Geometry and Physics, 24(3), pp.203-208.

\bibitem{Miura1976}
Miura, R.M., 1976. The Korteweg–deVries equation: a survey of results. 
SIAM review, 18(3), pp.412-459.


\bibitem{Poincare1901}
H. Poincar\'{e}, [1901],
Sur une forme nouvelle des \'{e}quations de la m\'{e}canique,
{\em C.R. Acad. Sci. Paris} {\bf 132}, 369--371.

\bibitem{Wazwaz2010}
Wazwaz, A.-M., 2010. Partial Differential Equations and Solitary Waves Theory (Springer, Berlin,Heidelberg.

\bibitem{Weinstein1996}
Weinstein, A. [1996],
Lagrangian mechanics and groupoids,
{\it Fields Inst. Comm.} {\bf 7}, 207--231.

\bibitem{Whitham[1974]} G.B. Whitham, {\it Linear and Nonlinear Waves},
Wiley Interscience:New York (1974).


\bibitem{LiZhiSibgatullin[1997]}
Li Zhi and N. R. Sibgatullin, J. Appl. Maths. Mechs. {\bf 61} 177 (1997)

\end{thebibliography}
\end{document}